\documentclass[3p,preprint]{elsarticle}
\usepackage{amsmath, amssymb, amsfonts, amsthm}
\usepackage{algorithmicx}
\usepackage{algorithm}
\usepackage{algpseudocode}
\usepackage{graphicx}
\usepackage[normalem]{ulem}
\usepackage{cancel}
\usepackage{natbib}
\usepackage{textcomp}
\usepackage{lineno}

\begin{document}
\begin{frontmatter}
\title{Finding~the~Leftmost~Critical~Factorization~on Unordered~Alphabet}
\author{Dmitry Kosolobov}
\ead{dkosolobov@mail.ru}
\address{Ural Federal University, Ekaterinburg, Russia}
\begin{abstract}
We present a linear time and space algorithm computing the leftmost critical factorization of a given string on an unordered alphabet.
\end{abstract}
\end{frontmatter}

\theoremstyle{plain}
\newtheorem{theorem}{Theorem}
\newtheorem{lemma}{Lemma}
\newtheorem*{example}{Example}

\algtext*{EndIf}
\algtext*{EndWhile}
\algtext*{EndFor}

\section{Introduction}

Stringology and combinatorics on words are closely related fields that intensively interact with each other. One of the most famous examples of their interaction is the surprising application of the so-called \emph{critical factorization}, a notion that was created inside the field of combinatorics on words for purely theoretic reasons (the precise definition is presented below). Critical factorizations are at the core of the constant space string matching algorithm by Crochemore and Perrin~\cite{CrochemorePerrin} and its real time variation by Breslauer, Grossi, and Mignosi~\cite{BreslauerGrossiMignosi}, which are, perhaps, the most elegant and simple string matching algorithms with such time and space bounds.

It is known that a critical factorization can be found in linear time and constant space when the input string is drawn from an ordered alphabet, i.e., when the alphabet is totally ordered and we can use symbol comparisons that test for the relative order of symbols (see~\cite{CrochemorePerrin,Duval}). In~\cite{BreslauerGrossiMignosi} it was posed as an open problem whether it is possible to find in linear time a critical factorization of a given string over an arbitrary unordered alphabet, i.e., when our algorithm is allowed to perform only equality comparisons. In this paper we answer this question affirmatively; namely, we describe a linear time algorithm finding the leftmost critical factorization of a given string on an unordered alphabet. A similar result is known for unbordered conjugates, a concept related to the critical factorizations: Duval et al.~\cite{DuvalLecroqLefebvre} proposed a linear algorithm that allows to find an unbordered conjugate of a given string on an arbitrary unordered alphabet. It is worth noting that all known so far algorithms working on general alphabets could find only \emph{some} critical factorization while our algorithm always finds the leftmost one. However, for the case of integer alphabet, there is a linear algorithm finding the leftmost critical factorization~\cite{DuvalEtAl} but it uses some structures (namely, the Lempel--Ziv decomposition) that cannot be computed in linear time on a general (even ordered) alphabet~\cite{Kosolobov}.

The paper is organized as follows. Section~\ref{SectPrel} contains some basic definitions and facts used throughout the text. In Section~\ref{SectAlg} we present our first algorithm and prove that its running time is $O(n\log n)$\footnote{For brevity, $\log$ denotes the logarithm with the base~$2$.} in Section~\ref{SectAnalysis}, where $n$ is the length of the input string. A more detailed analysis of this algorithm is given in Section~\ref{SectProblems}. In Section~\ref{SectLinearAlg} we improve our first solution to obtain a linear algorithm. Finally, we conclude with some remarks in Section~\ref{SectConclusion}.

\section{Preliminaries}\label{SectPrel}

We need the following basic definitions. A \emph{string $w$} over an alphabet $\Sigma$ is a map $\{1,2,\ldots,n\} \mapsto \Sigma$, where $n$ is referred to as the \emph{length of $w$}, denoted by $|w|$. We write $w[i]$ for the $i$th letter of $w$ and $w[i..j]$ for $w[i]w[i{+}1]\cdots w[j]$. Let $w[i..j]$ be the empty string for any~$i > j$. A string $u$ is a \emph{substring} (or a \emph{factor}) of $w$ if $u=w[i..j]$ for some $i$ and $j$. The pair $(i,j)$ is not necessarily unique; we say that $i$ specifies an \emph{occurrence} of $u$ in $w$. A string can have many occurrences in another string. A substring $w[1..j]$ [respectively, $w[i..n]$] is a \emph{prefix} [respectively, \emph{suffix}] of $w$. For integers $i$ and $j$, the set $\{k\in \mathbb{Z} \colon i \le k \le j\}$ (possibly empty) is denoted by $[i..j]$. Denote $[i..j) = [i..j{-}1]$, $(i..j] = [i{+}1..j]$, and $(i..j) = [i{+}1..j{-}1]$. Our notation for arrays is similar to that for strings: for example, $a[i..j]$ denotes an array indexed by the numbers $i, i{+}1, \ldots, j$.

Throughout the paper, we intensively use different periodic properties of strings. A string $u$ is called a \emph{border} of a string $w$ if $u$ is both a prefix and a suffix of $w$. A string is \emph{unbordered} if it has only trivial borders: the empty string and the string itself. An integer $p$ is a \emph{period} of $w$ if $0 < p \le |w|$ and $w[i] = w[i{+}p]$ for all $i=1,2,\ldots,|w|{-}p$. It is well known that $p > 0$ is a period of $w$ iff $w$ has a border of the length $|w| - p$. A string of the form $xx$, where $x$ is a nonempty string, is called a \emph{square}. Let $w[i..j] = xx$ for some $i$, $j$ and a nonempty string $x$; the position $i + |x|$ is called the \emph{center} of the square $w[i..j]$. A string $w$ is \emph{primitive} if $w \ne x^k$ for any string $x$ and any integer $k > 1$. A string $v$ is a \emph{conjugate} of a string $w$ if $v = w[i..|w|]w[1..i{-}1]$ for some $i$.
\begin{lemma}[see \cite{Lothaire}]
A string $w$ is primitive iff $w$ has an unbordered conjugate.\label{PrimitiveCrit}
\end{lemma}

\begin{figure}[htb]
\center
\includegraphics[scale=0.45]{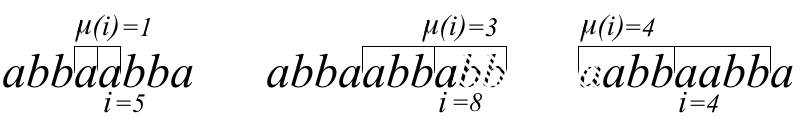}
\caption{Internal, right external, and left external local periods of the string $abbaabba$.}
\label{fig:locperiod}
\end{figure}
Now we can introduce the main notion of this paper. The \emph{local period} at a position $i$ (or centered at a position $i$) of $w$ is the minimal positive integer $\mu(i)$ such that the substring $w[\max\{1, i{-}\mu(i)\}..\min\{|w|, i{+}\mu(i){-}1\}]$ has the period $\mu(i)$ (see Figure~\ref{fig:locperiod}). Informally, the local period at a given position is the size of the smallest square centered at this position. We say that the local period $\mu(i)$ is \emph{left external} [respectively, \emph{right external}] if $i - \mu(i) < 1$ [respectively, $i + \mu(i) - 1 > |w|$]; the local period is \emph{external} if it is either left external or right external. The local period is \emph{internal} if it is not external. Obviously, the local period at any position of $w$ is less than or equal to the minimal period of $w$. A position $i$ of $w$ with the local period that is equal to the minimal period of $w$ is called a \emph{critical point}; the corresponding factorization $w[1..i{-}1]\cdot w[i..|w|]$ is called a \emph{critical factorization}. The following remarkable theorem holds.
\begin{theorem}[see~\cite{CesariVincent,Lothaire}]
Let $w$ be a string with the minimal period $p > 1$. Any sequence of $p{-}1$ consecutive positions of $w$ contains a critical point.\label{CritFact}
\end{theorem}

Theorem~\ref{CritFact} implies that any string with the minimal period $p$ has a critical point among the positions $1,2,\ldots, p$. Clearly, the local period corresponding to any such critical point is left external. The following lemmas are straightforward.
\begin{lemma}
If the local period at a position of a given string is both left external and right external, then this position is a critical point.\label{LeftRightExt}
\end{lemma}
\begin{lemma}
If the local period $\mu(i)$ at a position $i$ of a given string $w$ is not right external [respectively, left external], then the string $w[i..i{+}\mu(i){-}1]$ [respectively, $w[i{-}\mu(i)..i{-}1]$] is unbordered.\label{Unbordered}
\end{lemma}

\section{$O(n\log n)$ Algorithm}\label{SectAlg}

Our construction is based on the following observation.
\begin{lemma}
Let $w$ be a string with the minimal period $p > 1$. Denote $k = \max\{l \colon w[1..l] = w[j..j{+}l{-}1]\text{ for some }j\in (1..p]\}$. The leftmost critical point of $w$ is the leftmost position $i > k + 1$ with external local period.\label{LeftExtCrit}
\end{lemma}
\begin{proof}
Denote by $j$ a position such that $j \in (1..p]$ and $w[1..k] = w[j..j{+}k{-}1]$. Obviously, each of the positions $1,2,\ldots, k{+}1$ has the local period that is at most $j{-}1 < p$ (see Figure~\ref{fig:krange}) and hence cannot be a critical point.
\begin{figure}[htb]
\center
\includegraphics[scale=0.45]{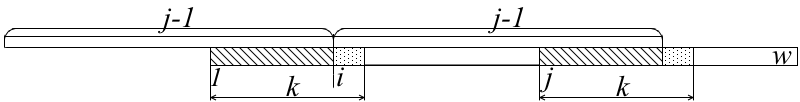}
\caption{The local period at a position $i \in [1..k{+}1]$.}
\label{fig:krange}
\end{figure}

Consider a position $i$ with left external local period $\mu(i) < p$. By Lemma~\ref{LeftRightExt}, $\mu(i)$ is not right external. So, we have $w[1..i{-}1] = w[\mu(i){+}1..i{+}\mu(i){-}1]$. Since $\mu(i) + 1 \le p$, by the definition of $k$, we have $i - 1 \le k$. Hence, any position $i > k + 1$ with left external local period is a critical point.

Now consider a position $i$ with right external local period $\mu(i) < p$. By Lemma~\ref{LeftRightExt}, $\mu(i)$ is not left external. It is easy to see that for any $i' \in (i..|w|]$, we have $\mu(i') \le \mu(i)$ and $i' - \mu(i') \ge 1$. Since Theorem~\ref{CritFact} implies that $w$ must have a critical point with left external local period, the position $i$ cannot be the leftmost position in $(k{+}1..|w|]$ with external local period.
\end{proof}

Hereafter, $w$ denotes the input string of length $n$ with the minimal period $p$. We process the trivial case $p = 1$ separately, so, assume $p > 1$. According to Theorem~\ref{CritFact} and Lemma~\ref{LeftExtCrit}, our algorithm processes only the first $p$ positions of $w$ from left to right starting from the position $k + 2$, where $k$ is defined as in Lemma~\ref{LeftExtCrit}, and when a local period at a given position $i$ is computed, then the following positions are skipped while they have at most the same local period. This leads to an $O(n \log n)$ time algorithm. To get a linear time algorithm, some local periods are reported from previous positions due to some local properties that are discussed in details in Section~\ref{SectLinearAlg}. More precisely, our $O(n\log n)$ algorithm is as follows.
\begin{algorithm}
\begin{algorithmic}[1]
\State compute $k = \max\{l \colon w[1..l] = w[j..j{+}l{-}1]\text{ for some }j\in (1..p]\}$\label{lst:k}
\State $i \gets k + 2;$
\While{$\mathbf{true}$}
    \State compute $\mu(i);$\label{lst:locper}
    \If{$\mu(i)$ is external}
        \State $i$ is the leftmost critical point; stop the algorithm;
    \EndIf
    \State $\mu \gets \mu(i);$\label{lst:skip0}
    \While{$w[i{-}1] = w[i{+}\mu{-}1]$}\label{lst:skip1}\Comment{skip positions that have local period at most $\mu$}
        \State $i \gets i + 1;$\label{lst:skip2}
    \EndWhile
\EndWhile
\caption{~}
\end{algorithmic}
\end{algorithm}

Obviously, the positions that the algorithm skips in lines~\ref{lst:skip1}--\ref{lst:skip2} have the local period at most $\mu < p$ and therefore cannot be critical points. So, Lemma~\ref{LeftExtCrit} immediately implies the correctness of Algorithm~1.

To calculate the number $k$ in $O(n)$ time, we utilize the following fact.
\begin{lemma}[{see~\cite[Chapter 1.5]{Gusfield}}]
For any strings $u$ and $w$, one can compute in $O(|u|)$ time an array $b[1..|u|]$ such that $b[j] = \max\{l \colon u[j..j{+}l{-}1] = w[1..l]\}$ for $j \in [1..|u|]$.\label{PartMatch}
\end{lemma}

To complete our construction, we describe an algorithm calculating the local period $\mu(i)$ at a given position $i$ provided $\mu(i)$ is internal. If this algorithm fails to compute $\mu(i)$, we decide that the local period is external.
\begin{lemma}
One can compute the internal local period $\mu(i)$ at a given position $i$ in $O(\mu(i))$ time and space.\label{LocPerCompute}
\end{lemma}
\begin{proof}
Fix an integer $x < i$. Let us first describe an algorithm that finds $\mu(i)$ in $O(x)$ time and space provided $\mu(i) \le x$. Using Lemma~\ref{PartMatch}, our algorithm constructs in $O(x)$ time an array $b[i{-}x..i{-}1]$ (for clarity, the indices start with $i{-}x$) of the length $x$ such that $b[j] = \max\{l \colon l \le x\text{ and }w[j..j{+}l{-}1] = w[i..i{+}l{-}1]\}$ for $j \in [i{-}x..i)$. It is straightforward that $\mu(i) = i - j$ for the rightmost $j \in [i{-}x..i)$ such that $b[j] \ge i - j$.

Now, to compute $\mu(i)$, we consecutively execute the above algorithm for $x = 2^0, 2^1, 2^2, \ldots, 2^{\lfloor\log(i-1)\rfloor}$ and, finally, for $x = i{-}1$ until we find $\mu(i)$. Thus, the algorithm runs in $O(\sum_{j=0}^{\lceil\log\mu(i)\rceil} 2^j) = O(\mu(i))$ time and space.
\end{proof}

\section{$O(n\log n)$ Time Bound}\label{SectAnalysis}

During the execution, Algorithm~1 calculates local periods at some positions. Let $S$ be the sequence of all such positions in the input string $w$ in increasing order. It is easy to see that the running time of the whole algorithm is $O(n + \sum_{i\in S}\mu(i))$. Thus, to prove that Algorithm~1 works in $O(n\log n)$ time, it suffices to show that $\sum_{i\in S} \mu(i) = O(n\log n)$. Simplifying the discussion, we exclude from $S$ all positions $i$ such that $\mu(i) = 1$.

Fix an arbitrary number $q$. Denote by $T(q)$ the maximal sum $\sum_{i\in S'} \mu(i)$ among all contiguous subsequences $S'$ of $S$ such that $\mu(i) \le q$ for each $i\in S'$. We are to show that $T(q) = O(q\log q)$, which immediately implies $\sum_{i\in S} \mu(i) = O(n\log n)$ since the number $q$ is arbitrary and $T(n) = \sum_{i\in S} \mu(i)$.

For further investigation, we need three additional combinatorial lemmas. Consider a position $i$ of $w$ with internal local period $\mu(i) > 1$. Informally, Lemma~\ref{LeftExt} shows that at the positions $(i..i{+}\mu(i))$ any internal local period that ``intersects'' the position $i$ and is not equal to $\mu(i)$ is either ``very short'' ($<\frac{1}2\mu(i)$) or ``very long'' ($\ge 2\mu(i)$). Lemma~\ref{Break} claims that always there is a ``long'' local period centered at $(i..i{+}\mu(i))$; moreover, this local period either is equal to $\mu(i)$ or is ``very long'' ($\ge 2\mu(i)$). Lemma~\ref{LeftExtBound} connects the bounds on the internal local periods that ``intersect'' the position $i$, as in Lemma~\ref{LeftExt}, and those local periods that do not ``intersect'' the position $i$. Now let us formulate these facts precisely.

\begin{lemma}
Let $i$ be a position of $w$ with internal local period $\mu(i) > 1$. For any $j \in (i..i{+}\mu(i))$ such that $j - \mu(j) < i$ and $\mu(j) \ne \mu(i)$, we have either $\mu(j) < \frac{1}2\mu(i)$ or $\mu(j) \ge 2\mu(i)$.\label{LeftExt}
\end{lemma}
\begin{proof}
The proof is essentially the same as in~\cite[Lemma 2]{ShurPetrova}. Let $\mu(j) \ge \frac{1}2\mu(i)$. Suppose $\mu(j) = \frac{1}2\mu(i)$. Since, by Lemma~\ref{Unbordered}, the string $w[i..i{+}\mu(i){-1}]$ is unbordered and hence cannot have the period $\mu(j) < \mu(i)$, we obtain $j + \mu(j) < i + \mu(i)$. The string $w[j{-}\mu(j)..j{+}\mu(j){-}1]$ is not primitive and has the length $\mu(i)$. Thus, the string $w[i..i{+}\mu(i){-}1]$ is a conjugate of $w[j{-}\mu(j)..j{+}\mu(j){-}1]$ and therefore is not primitive, a contradiction.

\begin{figure}[htb]
\center
\small a\includegraphics[scale=0.55,clip,trim=0 0 0 0]{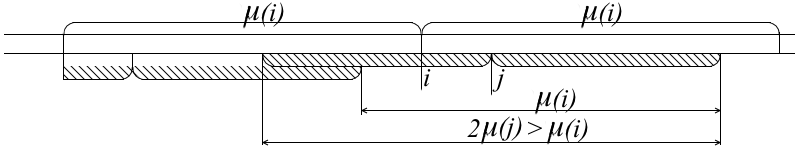}
\small b\includegraphics[scale=0.55,clip,trim=0 0 0 0]{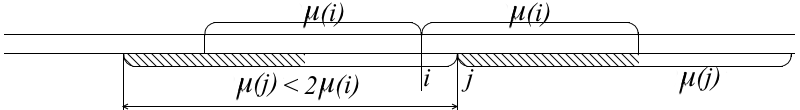}
\caption{Two impossible cases in Lemma~\ref{LeftExt}: (a) $\mu(i)/2 < \mu(j) < \mu(i)$, (b) $\mu(i) < \mu(j) < 2\mu(i)$.}
\label{fig:smallperiod}
\end{figure}
Now suppose $\mu(i)/2 < \mu(j) < \mu(i)$. As above, we have $j + \mu(j) < i + \mu(i)$. Thus, the string $w[j..j{+}\mu(j){-}1]$ has an occurrence $w[j{-}\mu(i)..j{-}\mu(i){+}\mu(j){-}1]$ that overlaps the string $w[j{-}\mu(j)..j{-}1] = w[j..j{+}\mu(j){-}1]$ because $2\mu(j) > \mu(i)$ (see Figure~\ref{fig:smallperiod} a). But, by Lemma~\ref{Unbordered}, $w[j{-}\mu(j)..j{-}1]$ is unbordered and therefore cannot overlap its own copy. This is a contradiction.

Finally, suppose $\mu(j) > \mu(i)$. By Lemma~\ref{Unbordered}, $w[j{-}\mu(j)..j{-}1]$ is unbordered. If $j - \mu(j) \ge i - \mu(i)$, then $w[j{-}\mu(j)..j{-}1]$ has the period $\mu(i) < \mu(j)$, a contradiction. Hence, we have $j - \mu(j) < i - \mu(i)$. If $\mu(j) < 2\mu(i)$, then the string $w[j..i{+}\mu(i){-}1]$, which is a suffix of $w[i..i{+}\mu(i){-}1]$, has an occurrence $w[j{-}\mu(j)..i{+}\mu(i){-}\mu(j){-}1]$ that overlaps $w[i{-}\mu(i)..i{-}1] = w[i..i{+}\mu(i){-}1]$ (see Figure~\ref{fig:smallperiod} b). This is a contradiction because, by Lemma~\ref{Unbordered}, $w[i{-}\mu(i)..i{-}1]$ is unbordered.
\end{proof}

\begin{lemma}
Let $i$ be a position of $w$ with internal local period $\mu(i) > 1$. Then there exists $j \in (i..i{+}\mu(i))$ such that either $\mu(j) = \mu(i)$ or $\mu(j) \ge 2\mu(i)$.\label{Break}
\end{lemma}
\begin{proof}
By Lemma~\ref{Unbordered}, the string $w[i..i{+}\mu(i){-}1]$ is unbordered and its minimal period is $\mu(i)$. For any position $j \in (i..i{+}\mu(i))$, denote by $\mu'(j)$ the local period in $j$ with respect to the substring $w[i..i{+}\mu(i){-}1]$. Observe that $\mu'(j) \le \mu(j)$. By Theorem~\ref{CritFact}, there is $j \in (i..i{+}\mu(i))$ such that $\mu'(j) = \mu(i)$ and $j - \mu'(j) < i$. Hence, we have $\mu(j) \ge \mu(i)$ and, moreover, if $\mu(j) > \mu(i)$, then, by Lemma~\ref{LeftExt}, $\mu(j) \ge 2\mu(i)$.
\end{proof}

\begin{lemma}
Let $i$ be a position of $w$ with internal local period $\mu(i) > 1$. Fix $j \in (i..i{+}\mu(i))$. Then, for any $h \in (i..j]$ such that $\mu(h) > 1$, we have $\mu(h) \le \max\{\mu(h') \colon h' \in (i..j]\text{ and }h' - \mu(h') < i\}$.\label{LeftExtBound}
\end{lemma}
\begin{proof}
Suppose, to the contrary, there is $h \in (i..j]$ such that $\mu(h) > 1$ and $\mu(h) > \max\{\mu(h') \colon h' \in (i..j]\text{ and }h' - \mu(h') < i\}$; let $h$ be the leftmost such position. Then, we have $h - \mu(h) \ge i$. Using a symmetrical version of Lemma~\ref{Break}, we obtain $h' \in (h{-}\mu(h)..h)$ such that $\mu(h') \ge \mu(h)$. Since $\mu(h') \ge \mu(h)$, by the definition of $h$, we have $h' - \mu(h') \ge i$. This contradicts to the choice of $h$ as the leftmost position with the given properties because $h' < h$ and $h' \in (i..j]$.
\end{proof}

Hereafter, $S' = \{i_1, i_2, \ldots, i_z\}$ denotes a contiguous subsequence of $S$ such that $\mu(i_j) \le q$ for each $j\in [1..z]$ and $T(q) = \sum_{j=1}^z \mu(i_j)$. We associate with each $i_j$ the numbers $r_j = \max\{r \colon w[i_j{-}\mu(i_j)..r{-}1]$ has the period $\mu(i_j)\}$ and $c_j = \max\{c \le r_j{-}\mu(i_j) \colon w[c..c{+}\mu(i_j){-}1]$ is unbordered$\}$ (see Figure~\ref{fig:icr}). By Lemma~\ref{Unbordered}, the string $w[i_j..i_j{+}\mu(i_j){-}1]$ is unbordered and therefore $c_j \ge i_j$. Since $w[c_j..c_j{+}\mu(i_j){-}1]$ is unbordered and $w[c_j{-}\mu(i_j)..c_j{-}1] = w[c_j..c_j{+}\mu(i_j){-}1]$, we have $\mu(c_j) = \mu(i_j)$. Since $w[r_j{-}\mu(i_j)..r_j{-}1]$ is primitive, it follows from Lemma~\ref{PrimitiveCrit} that $c_j > r_j - 2\mu(i_j)$. Algorithm~1 skips the positions $i_j + 1, i_j + 2, \ldots, r_j - \mu(i_j)$ in the loop in lines~\ref{lst:skip1}--\ref{lst:skip2}.
\begin{figure}[htb]
\center
\includegraphics[scale=0.55]{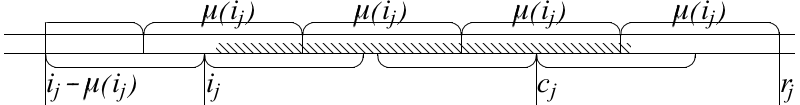}
\caption{The positions $i_j{+}1, i_j{+}2, \ldots, r_j{-}\mu(i_j)$ are shaded.}
\label{fig:icr}
\end{figure}
\begin{lemma}
For any $j\in [1..z]$ and $i \in (c_j..c_j{+}\mu(c_j))$, we have $\mu(i) \ne \mu(c_j)$.\label{cjump}
\end{lemma}
\begin{proof}
For converse, suppose $\mu(i) = \mu(c_j)$. Since $w[i{-}\mu(i)..i{-}1] = w[i..i{+}\mu(i){-}1]$ and $\mu(i) = \mu(c_j) = \mu(i_j)$, by the definition of $r_j$, we have $i \le r_j - \mu(i_j)$. It follows from Lemma~\ref{Unbordered} that $w[i..i{+}\mu(i){-}1]$ is unbordered. This contradicts to the definition of $c_j$ because $c_j < i \le r_j - \mu(i_j)$.
\end{proof}

To estimate the sum $\sum_{j=1}^z \mu(i_j)$, we construct a subsequence $i_{s_1}, i_{s_2}, \ldots, i_{s_t}$ by the following inductive process. Choose $i_{s_1} = i_1$. Suppose we have already constructed a subsequence $i_{s_1}, i_{s_2}, \ldots, i_{s_j}$. Choose the minimal number $i' \in (c_{s_j}..c_{s_j}{+}\mu(c_{s_j}))$ such that $\mu(i') \ge \mu(c_{s_j})$. By Lemma~\ref{Break}, such number always exists. If $i' > i_z$, we set $t = j$ and stop the process. Let $i' \le i_z$. It follows from Lemma~\ref{cjump} that $\mu(i') \ne \mu(c_{s_j})$. Hence, by Lemma~\ref{Break}, $\mu(i') \ge 2\mu(c_{s_j}) = 2\mu(i_{s_j})$. Since $\mu(i') > \mu(i_{s_j})$, it follows from the definition of $r_{s_j}$ that $i' > r_{s_j} - \mu(i_{s_j})$. Therefore, Algorithm~1 does not skip $i'$ and $i' \in S$. Since $\{i_1, i_2, \ldots, i_z\}$ is a contiguous subsequence of $S$, we have $i' = i_{j'}$ for some $j' \in [1..z]$. Set $i_{s_{j+1}} = i_{j'}$.

Now we can prove that the running time of Algorithm~1 is $O(n\log n)$. For any $j\in [1..t)$, we have $\mu(i_{s_{j+1}}) \ge 2\mu(i_{s_j})$ and therefore $\sum_{j=1}^t \mu(i_{s_j}) \le \mu(i_{s_t}) + \frac{1}2\mu(i_{s_t}) + \frac{1}{2^2}\mu(i_{s_t}) + \cdots \le 2\mu(i_{s_t}) \le 2q$. Further, let $h \in [1..z]$ and $i_{s_j} < i_h < i_{s_{j+1}}$ for some $j \in [1..t)$. Since Algorithm~1 skips the positions $(i_{s_j}..c_{s_j}]$ and $i_{s_{j+1}} \in (c_{s_j}..c_{s_j}{+}\mu(c_{s_j}))$, it follows that $i_h \in (c_{s_j}..c_{s_j}{+}\mu(c_{s_j}))$. Recall that $i_{s_{j+1}}$ is the minimal number from $(c_{s_j}..c_{s_j}{+}\mu(c_{s_j}))$ such that $\mu(i_{s_{j+1}}) \ge \mu(c_{s_j})$. Thus, by Lemmas~\ref{LeftExt} and~\ref{LeftExtBound}, we have $\mu(i_h) < \frac{1}2 \mu(c_{s_j}) = \frac{1}2 \mu(i_{s_j})$. In the same way, for $h \in [1..z]$ such that $i_h > i_{s_t}$, we have $\mu(i_h) < \frac{1}2 \mu(i_{s_t})$. So, we obtain the following recursion:
\begin{equation}
T(q) \le 2q + T\left(\frac{1}2 \mu(i_{s_1})\right) + T\left(\frac{1}2 \mu(i_{s_2})\right) + \cdots + T\left(\frac{1}2 \mu(i_{s_t})\right)\enspace.
\label{eq:teq}
\end{equation}
Consider a recursion $T(q) = O(q) + \sum_{j=1}^t T(q_j)$. It is well known that if the sum of the terms from the parentheses of $T(\ldots)$ in the right hand side of this recursion (i.e., $\sum_{j=1}^t q_j$) is less than or equal to $q$ and each of those terms (i.e., each $q_j$) is less than or equal to $\frac{1}2q$, then the recursion has a solution $T(q) = O(q\log q)$. Thus, since the sum of the terms from the parentheses of $T(\ldots)$ in the right hand side of~(\ref{eq:teq}) is equal to $\frac{1}2\sum_{j=1}^t \mu(i_{s_j}) \le q$ and each of these terms is less than or equal to $\frac{1}2 q$, we obtain $T(q) = O(q\log q)$.

\section{Problems with Linearity}\label{SectProblems}

To obtain $T(q) = O(q)$, we might prove that if $2\mu(i_{s_{t-1}})$ and $\mu(i_{s_t})$ are close enough (namely, $\frac{7}3 \mu(i_{s_{t-1}}) > \mu(i_{s_t})$), the term $T(\frac{1}2 \mu(i_{s_t}))$ in (\ref{eq:teq}) is actually $T(\frac{2}3\mu(i_{s_{t-1}})) \le T(\frac{1}3\mu(i_{s_t}))$; this fact would imply that the sum of the terms in the parentheses of $T(\ldots)$ in the right hand side of (\ref{eq:teq}) is less than $\alpha q$ for some constant $\alpha < 1$ and therefore $T(q) = O(q)$. Unfortunately, this is not true for Algorithm~1. Nevertheless, we prove a restricted version of the mentioned claim. It reveals problems that may arise in the current solution and points out a way to improvements.

\begin{lemma}
Let $i\in (c_{s_t}..c_{s_t}{+}\mu(c_{s_t}))$. Suppose $\mu(i') < \mu(c_{s_t})$ and $\mu(i') \ne \mu(i_{s_{t-1}})$ for each $i' \in (c_{s_t}..i]$. If $\frac{7}3 \mu(i_{s_{t-1}}) > \mu(i_{s_t})$, then $\mu(i) < \frac{2}3 \mu(i_{s_{t-1}})$.\label{MainLemma}
\end{lemma}
\begin{proof}
Recall that $2\mu(c_{s_{t-1}}) \le \mu(i_{s_t})$. Denote $a = w[c_{s_{t-1}} .. c_{s_{t-1}}{+}\mu(c_{s_{t-1}}){-}1]$ and $b = w[c_{s_{t-1}}{+}\mu(c_{s_{t-1}}) .. c_{s_{t-1}}{-}\mu(c_{s_{t-1}}){+}\mu(i_{s_t}){-}1]$ (see Figure~\ref{fig:notation}). Note that $\mu(c_{s_{t-1}}) = |a|$ and $\mu(c_{s_t}) = |aab|$. It follows from Lemma~\ref{Unbordered} that $a$ is unbordered. Since, by Lemma~\ref{Unbordered}, the string $w[i_{s_t}..i_{s_t}{+}\mu(i_{s_t}){-}1]$ is unbordered, the string $b$ is not empty. The inequality $\frac{7}3|a| = \frac{7}3 \mu(i_{s_{t-1}}) > \mu(i_{s_t}) = |baa|$ implies $|b| < \frac{1}3|a|$.

\begin{figure}[htb]
\center
\includegraphics[scale=0.55]{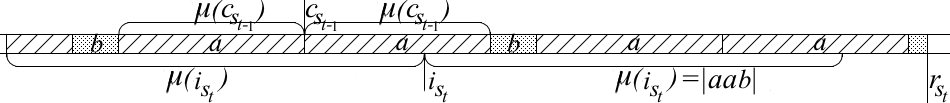}
\caption{The strings $a$ and $b$.}
\label{fig:notation}
\end{figure}

In view of Lemma~\ref{LeftExtBound}, it suffices to prove the lemma only for the positions $i$ such that $i - \mu(i) < c_{s_t}$. So, assume $i - \mu(i) < c_{s_t}$. Since $\mu(i) < \mu(c_{s_t})$, it follows from Lemma~\ref{LeftExt} that $\mu(i) < \frac{1}2 \mu(c_{s_t}) = \frac{1}2 |baa| < |ab|$. Since, by Lemma~\ref{Unbordered}, $w[c_{s_t}..c_{s_t}{+}\mu(c_{s_t}){-}1]$ is unbordered and thus cannot have the period $\mu(i) < \mu(c_{s_t})$, we obtain $i + \mu(i) < c_{s_t} + \mu(c_{s_t})$. So, $w[i{-}\mu(i)..i{+}\mu(i){-}1]$ is a substring of the string $w[i_{s_t}{-}\mu(i_{s_t})..r_{s_t}{-}1]$. Therefore, since $w[i_{s_t}{-}\mu(i_{s_t})..r_{s_t}{-}1]$ has the period $\mu(i_{s_t}) = \mu(c_{s_t}) = |aab|$, the string $w[i{-}\mu(i)..i{+}\mu(i){-}1]$ is a substring of the string $u = aabaabaab$ (see Figure~\ref{fig:notation}). Thus, to finish the proof, it suffices to prove the following claim. 

\emph{Claim. Let $i$ be a position of $u$ with internal local period $\mu(i)$ (the local period at $i$ is with respect to the string $u$). If $\mu(i) < |ab|$ and $\mu(i) \ne |a|$, then $\mu(i) < \frac{2}3 |a|$.}

Let $i$ be a position of $u$ with internal local period $\mu(i)$ such that $\mu(i) < |ab|$ and $\mu(i) \ne |a|$. Consider two cases.

1) Suppose $i$ lies in an occurrence of $a$ in $u = aabaabaab$. Without loss of generality, consider the case $i \in (|aaba|..|aabaa|]$; all other cases are similar. If $i - \mu(i) \le |aaba|$, then, by Lemma~\ref{LeftExt}, we have either $\mu(i) < \frac{1}2|a|$ or $\mu(i) \ge 2|a|$. The latter is impossible because $\mu(i) < |ab| < 2|a|$ while the former implies $\mu(i) < \frac{2}3|a|$ as required. Now let $i - \mu(i) > |aaba|$. Assume, by a contradiction, that $\mu(i) \ge \frac{2}3|a|$. Then $w[i{-}\mu(i)..i{-}1]$ is a substring of $a$ and thus it has an occurrence $v = w[i{-}\mu(i){+}|ab|..i{-}1{+}|ab|]$ (see Figure~\ref{fig:centera}). Since $2\mu(i) \ge \frac{4}3|a| > |ab|$, the string $w[i..i{+}\mu(i){-}1]$, which is also an occurrence of $w[i{-}\mu(i)..i{-}1]$, overlaps $v$. This is a contradiction because $w[i{-}\mu(i)..i{-}1]$ is unbordered by Lemma~\ref{Unbordered}.
\begin{figure}[htb]
\center
\includegraphics[scale=0.55]{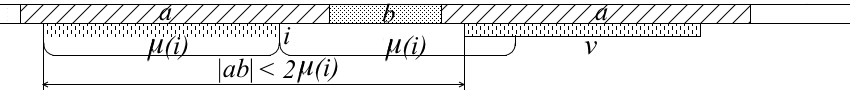}
\caption{The impossible case $i \in (|aaba|..|aabaa|]$ and $i - \mu(i) > |aaba|$ from the proof of Lemma~\ref{MainLemma}.}
\label{fig:centera}
\end{figure}

2) Suppose $i$ lies in an occurrence of $b$ in $u = aabaabaab$. Without loss of generality, consider the case $i \in (|aa|..|aab|]$. Assume, by a contradiction, that $\mu(i) \ge \frac{2}3|a|$. Suppose $i - \mu(i) > |a|$ (see Figure~\ref{fig:centerc}a). Then the string $w[i{-}\mu(i)..|aa|]$, which is a suffix of $a$, has an occurrence $v = w[i..|aa|{+}\mu(i)]$. Since $\mu(i) \ge \frac{2}3|a| > |b|$, $v$ overlaps $w[|aab|{+}1..|aaba|] = a$. Hence, $a$ has a nontrivial border, clearly a contradiction. Suppose $i - \mu(i) \le |a|$ (see Figure~\ref{fig:centerc}b). Then the string $w[|a|{+}1..|aa|] = a$ has an occurrence $v = w[|a|{+}1{+}\mu(i)..|aa|{+}\mu(i)]$. Since $\mu(i) < |ab|$ and $\mu(i) + |a| \ge \frac{5}3|a| > |ab|$, the string $w[|aab|{+}1..|aaba|] = a$ overlaps $v = a$. This is a contradiction because $a$ is unbordered.
\begin{figure}[htb]
\center
\small a\includegraphics[scale=0.55]{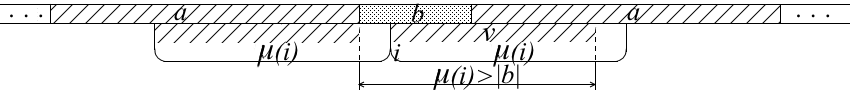}
\small b\includegraphics[scale=0.55]{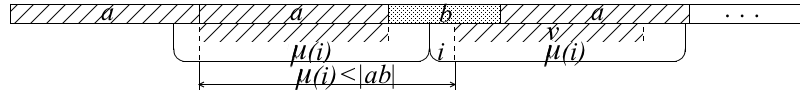}
\caption{The impossible cases for $i \in (|aa|..|aab|]$ in the proof of Lemma~\ref{MainLemma}: (a) $i - \mu(i) > |a|$; (b) $i - \mu(i) \le |a|$.}
\label{fig:centerc}
\end{figure}
\end{proof}

Let us consider how one might use Lemma~\ref{MainLemma} to obtain $T(q) = O(q)$. Suppose $t > 1$, $\frac{7}3\mu(i_{s_{t-1}}) > \mu(i_{s_t})$, and $\mu(i_h) \ne \mu(i_{s_{t-1}})$ for all $h \in (s_t..z]$. Lemma~\ref{MainLemma} implies that $\mu(i_h) < \frac{2}3\mu(i_{s_{t-1}}) \le \frac{1}3 \mu(i_{s_t})$ for each $h \in (s_t..z]$. So, combining Lemmas~\ref{LeftExt},~\ref{LeftExtBound},~\ref{MainLemma}, one can deduce the following recursion:
\begin{equation}
T(q) \le \sum_{j=1}^t \mu(i_{s_j}) + T\left(\frac{1}2\mu(i_{s_1})\right) + \cdots + T\left(\frac{1}2\mu(i_{s_{t-1}})\right) + T\left(\frac{1}3\mu(i_{s_t})\right)\enspace.
\label{eq:early_recustion}
\end{equation}
Let us estimate the sum of the terms from the parentheses of $T(\ldots)$ in the right hand side of (\ref{eq:early_recustion}). Since $\sum_{j=1}^{t-1} \mu(i_{s_j}) \le q$, we have $\frac{1}2 \mu(i_{s_1}) + \cdots + \frac{1}2 \mu(i_{s_{t-1}}) + \frac{1}3 \mu(i_{s_t}) \le \frac{1}2q + \frac{1}3 q = \frac{5}6 q$. The sum $\sum_{j=1}^t \mu(i_{s_j})$ is bounded by~$2q$. It is well known that such recursion has a solution $T(q) \le 2q + \frac{5}6 2q + (\frac{5}6)^2 2q + \cdots = O(q)$. Unfortunately, a fatal problem arises when there is $h \in (s_t..z]$ such that $\mu(i_h) = \mu(i_{s_{t-1}})$. Exploiting this case, we construct a string on which Algorithm~1 performs $\Omega(n\log n)$ operations.

\vskip3mm\noindent\textbf{Example.} Let $a_i$ and $b_i$ be sequences of strings inductively defined as follows: $a_0 = a$, $b_0 = b$ and $a_{i+1} = a_i\$_ia_i$, $b_{i+1} = b_ia_i\$_ia_ib_i$, where $a, b, \$_0, \$_1, \$_2, \ldots$ are distinct letters. Denote $w_i = a_ib_ia_i$. Note that $w_{i+1} = a_i\$_iw_i\$_iw_i\$_ia_i$; this recursive structure of $w_{i+1}$ is very important for us. Our counterexample is the string $w = \#w_{i+1}\#a_{i+1}\#$, where $\#$ is a unique special letter. Clearly, the minimal period of $w$ is $|w|{-}1$. Since $w = \#a_{i+1}b_{i+1}a_{i+1}\#a_{i+1}\#$, it is easy to see that the number $k = \max\{l \colon w[1..l] = w[j..j{+}l{-}1]\text{ for some }j\in (1..|w|)\}$ is equal to $|\#a_{i+1}|$. So, Algorithm~1 starts with the position $|\#a_{i+1}|{+}2$. Now consider some combinatorial properties of $w_i$.

\begin{lemma}
The string $w_i = a_ib_ia_i$ satisfies the following conditions:\\
(1) the local period at each of the positions $[|a_i|{+}2..|a_ib_i|]$ is internal;\\
(2) the local period at position $|a_ib_i|{+}1$ is right external.\label{ExampleLemma}
\end{lemma}
\begin{proof}
The proof is by induction on $i$. The base case $w_0 = aba$ is obvious. The inductive step is $w_{i+1} = a_{i+1}b_{i+1}a_{i+1} = a_i\$_ia_i\cdot b_ia_i\$_ia_ib_i\cdot a_i\$_ia_i = a_i\$_iw_i\$_iw_i\$_ia_i$. Consider condition~(1). The positions $[|a_{i+1}|{+}2..|a_{i+1}b_i|]$ correspond to the positions $[|a_i|{+}2..|a_ib_i|]$ of the first occurrence of the string $w_i = a_ib_ia_i$ in $w_{i+1}$. Hence, by the inductive hypothesis, the local periods at these positions are internal. It is obvious that $p = |a_i\$_ia_ib_i|$ is a period of $w_{i+1}$ and therefore the positions $(p..|w|{-}p{+}1]$ all have internal local periods. So, it suffices to consider the positions $[|w|{-}p{+}2..|a_{i+1}b_{i+1}|] = [|a_{i+1}b_ia_i\$_ia_i|{+}2..|a_{i+1}b_{i+1}|]$. Similarly, these positions correspond to the positions $[|a_i|{+}2..|a_ib_i|]$ of the second occurrence of the substring $w_i = a_ib_ia_i$ in $w$. Therefore, by the inductive hypothesis, all these positions have internal local periods. Consider condition~(2). Denote $j = |a_{i+1}b_{i+1}{+}1|$. By the inductive hypothesis, $\mu(j) > |a_i|$. Now since $w[j{+}|a_i|] = \$_i$, it is easy to see that $\mu(j) > |a_{i+1}|$, i.e., $\mu(j)$ is right external.
\end{proof}

The main loop of Algorithm~1 starts with the position $|\#a_{i+1}|{+}2 = |a_i\$_ia_i|{+}2$, i.e., with the position $|a_i|{+}2$ inside the first occurrence of $w_i$ in $w_{i+1} = a_i\$_iw_i\$_iw_i\$_ia_i$. By Lemma~\ref{ExampleLemma}, we process $w_i$ until the position $|a_ib_i|{+}1$ in $w_i$ that corresponds to the position $j = |\#a_i\$_ia_ib_i|{+}1$ in $w$ is reached. By Lemma~\ref{ExampleLemma}, we have $\mu(j) > |a_i|$. Hence, it is straightforward that $\mu(j) = |a_i\$_ia_ib_i|$, which is a period of the whole string $w_{i+1}$. Algorithm~1 calculates $\mu(j)$ and then skips some positions in the loop in lines~\ref{lst:skip1}--\ref{lst:skip2} until it reaches the position $j' = |\#a_i\$_iw_i\$_ia_i|{+}2$, all in $\Theta(|w_{i+1}|)$ time. The position $j'$ corresponds to the position $|a_i|{+}2$ inside the second occurrence of $w_i$ in $w_{i+1} = a_i\$_iw_i\$_iw_i\$_ia_i$. So, we have some kind of recursion here. Denote by $t_{i+1}$ the time required to process the substring $w_{i+1}$ of $w$; it follows from our discussion that $t_{i+1}$ can be expressed by the following recursive formula: $t_{i+1} = \Theta(|w_{i+1}|) + 2t_i$ (with $t_0 = 0$). For simplicity, assume that the constant under the $\Theta$ is~$1$, so, $t_{i+1} = |w_{i+1}| + 2t_i$.

To estimate $t_{i+1}$, we first solve the following recursions: $|a_{i+1}| = 2|a_i| + 1$, $|b_{i+1}| = 2|b_i| + 2|a_i| + 1$, $|w_i| = 2|a_i| + |b_i|$ (with $|a_0| = |b_0| = 1$). Obviously $|a_i| = 2^{i+1} - 1$. Then $|b_{i+1}| = 2^{i+2} - 1 + 2|b_i|$. By a simple substitution, one can show that $|b_i| = i2^{i+1} + 1$. So, we obtain $|w_i| = i2^{i+1} + 2^{i+2} - 1$ and therefore $t_i = i2^{i+1} + 2^{i+2} - 1 + 2t_{i-1}$. By a substitution, one can prove that $t_i = i^22^i + 5i2^i - 2^i + 1$: indeed, substituting $t_{i-1} = (i-1)^22^{i-1} + 5(i-1)2^{i-1} - 2^{i-1} + 1$, we obtain
$$
\begin{array}{l}
t_i = i2^{i+1} + 2^{i+2} - 1 + 2t_{i-1}\\
 = \uline{i2^{i+1}} + \uwave{2^{i+2}} - 1 + (\uuline{(i-1)^22^i} + \uline{5(i}\uwave{-1)2^i} - \uwave{2^i} + 2)\\
 = \uuline{i^22^i - \cancel{2i2^i} + \cancel{2^i}} + \uline{\cancel{i2^{i+1}}} + \uline{5i2^i} + \uwave{2^{i+2}} - \uwave{5\cdot 2^i} - \uwave{\cancel{2^i}} + 1\\
 = \uuline{i^22^i} + \uline{5i2^i} - \uwave{2^i} + 1\enspace.
\end{array}
$$
Finally, since $|w_{i+1}| = (i+1)2^{i+2} + 2^{i+3} - 1 = \Theta(i2^i)$ and $\log|w_{i+1}| = \Theta(i)$, we obtain $t_{i+1} = (i + 1)^22^{i+1} + 5(i+1)2^{i+1} - 2^{i+1} + 1 = \Theta(i^22^i) = \Theta(|w_{i+1}|\log |w_{i+1}|) = \Theta(|w|\log|w|)$.

\section{Linear Algorithm}\label{SectLinearAlg}

To overcome the issues addressed in the previous section, we introduce two auxiliary arrays $m[1..n]$ and $r[1..n]$ that are initially filled with zeros; their meaning is clarified by Lemma~\ref{Correct} below. In Algorithm~2 below we use the three-operand $\mathbf{for}$ loop like in the C language.
\begin{algorithm}
\begin{algorithmic}[1]
\State compute $k = \max\{l \colon w[1..l] = w[j..j{+}l{-}1]\text{ for some }j\in (1..p]\}$
\State $i \gets k + 2;$
\While{$\mathbf{true}$}
    \If{$m[i] = 0$}\label{lst:cond} \Comment{$m[i]$ is not computed}
        \State compute $\mu(i);$\label{lst:computelocper}
        \If{$\mu(i)$ is external}
            \State $i$ is the leftmost critical point; stop the algorithm;%
        \EndIf%
        \State $m[i] \gets \mu(i);$
        \State $r[i] \gets i + m[i];$
        \While{$w[r[i]{-}m[i]] = w[r[i]]$}\label{lst:copybeg}
            \State $r[i] \gets r[i] + 1;$\label{lst:extendr}
        \EndWhile
        \For{$(j\gets i - m[i];\; j < r[i]{-}m[i];\; j \gets j + 1)$}\label{lst:copy0}
            \If{$m[j] \ne 0 \mathrel{\mathbf{and}} j - m[j] \ge i - m[i] \mathrel{\mathbf{and}} r[j] + m[i] < r[i]$}\label{lst:copy1}%
                \State $m[j{+}m[i]] \gets m[j];$\label{lst:copy21}
                \State $r[j{+}m[i]] \gets r[j] + m[i];$\label{lst:copy22}
            \EndIf
        \EndFor\label{lst:copyend}%
    \EndIf
    \State $i \gets r[i] - m[i] + 1;$\label{lst:xskipShort}
\EndWhile
\end{algorithmic}
\caption{~}
\end{algorithm}

\begin{lemma}
If $m[i] \ne 0$ for some position $i$ during the execution of Algorithm~2, then $m[i] = \mu(i)$ and $r[i] = \max\{r \colon w[i..r{-}1]\text{ has the period }\mu(i)\}$.\label{Correct}
\end{lemma}
\begin{proof}
For each position $j$, denote $r_j = \max\{r \colon w[j..r{-}1]\text{ has the period }\mu(j)\}$. It suffices to show that the assignments in lines~\ref{lst:copy21}--\ref{lst:copy22} always assign $\mu(j{+}m[i])$ to $m[j{+}m[i]]$ and $r_{j{+}m[i]}$ to $r[j{+}m[i]]$. Suppose Algorithm~2 performs line~\ref{lst:copy21} for some $j$. Evidently, the string $w[i{-}m[i]..r[i]{-}1]$ has the period $m[i]$ (see Figure~\ref{fig:copyj}). Further, by the condition in line~\ref{lst:copy1}, the strings $w[j{-}m[j]..r[j]]$ and $w[j{-}m[j]{+}m[i]..r[j]{+}m[i]]$ are substrings of $w[i{-}m[i]..r[i]{-}1]$ and therefore they are equal. Hence, we have $\mu(j) = \mu(j{+}m[i])$ and $r_j + m[i] = r_{j{+}m[i]}$ provided $\mu(j) = m[j]$ and $r_j = r[j]$. Now one can prove the desired claim by a simple induction.
\begin{figure}[htb]
\center
\includegraphics[scale=0.55]{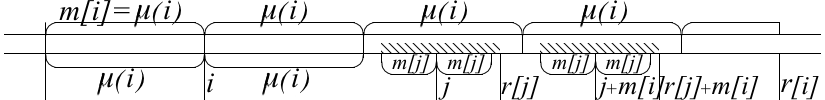}
\caption{$j - m[j] \ge i - m[i]$ and $r[j] + m[i] < r[i]$.}
\label{fig:copyj}
\end{figure}
\end{proof}
By Lemma~\ref{Correct}, the assignment in line~\ref{lst:xskipShort} skips exactly the same set of positions as the loop in lines~\ref{lst:skip0}--\ref{lst:skip2} in Algorithm~1. Thus, Lemma~\ref{Correct} implies that the values $m[i] = \mu(i)$ computed by Algorithm~2 coincide with the same values computed by Algorithm~1 and hence are correct. However, now we do not compute some local periods but copy them from the array $m$ instead. It turns out that this is crucial for the time analysis.

As above, let $S$ be the sequence of all positions that Algorithm~2 does not skip in line~\ref{lst:xskipShort}. Again, we exclude from $S$ all positions $i$ such that $\mu(i) = 1$. Evidently, the resulting sequence is exactly the same as the sequence $S$ in Section~\ref{SectAnalysis} but, in contrast to Algorithm~1, the new algorithm copies local periods at some positions of $S$ from the array $m$ rather than calculates them explicitly. Denote by $\hat{S}$ the subsequence of all positions of $S$ for which Algorithm~2 computes local periods explicitly in line~\ref{lst:computelocper}.

Due to the assignment in line~\ref{lst:xskipShort}, obviously, the loop in lines~\ref{lst:copybeg}--\ref{lst:extendr} performs at most $n$ iterations in total.
The loop in lines~\ref{lst:copy0}--\ref{lst:copyend} performs exactly the same number of iterations as the loop in lines~\ref{lst:copybeg}--\ref{lst:extendr} plus $\mu(i)$ iterations for an appropriate $i \in \hat{S}$. Hence, the running time of the whole algorithm is $O(n + \sum_{i\in \hat{S}}\mu(i))$. Thus, to prove that Algorithm~2 is linear, it suffices to show that $\sum_{i\in \hat{S}} \mu(i) = O(n)$.

Fix an arbitrary number $q$. Denote by $T(q)$ the maximal sum $\sum_{i\in S' \cap \hat{S}} \mu(i)$ among all contiguous subsequences $S'$ of $S$ such that $\mu(i) \le q$ for each $i\in S'$ (note that we sum only through the positions of $\hat{S}$). We are to show that $T(q) = O(q)$, which immediately implies $\sum_{i \in \hat{S}} \mu(i) = O(n)$ since the number $q$ is arbitrary and $T(n) = \sum_{i\in \hat{S}} \mu(i)$.

We need one additional combinatorial fact.
\begin{lemma}
Let $i$ be a position of $w$ with internal local period $\mu(i) > 1$. Suppose $j$ is a position from $(i..i{+}\mu(i))$ such that $\mu(j') < \mu(i)$ for each $j' \in (i..j]$; then $w[j{-}\mu(j)..j{+}\mu(j){-}1]$ is a substring of $w[i{-}\mu(i)..i{+}\mu(i){-}1]$.
\label{NotRight}
\end{lemma}
\begin{proof}
Assume, by a contradiction, that $j + \mu(j) > i + \mu(i)$. For each $h \in [i..i{+}\mu(i))$, denote by $\mu'(h)$ the local period at the position $h$ with respect to the substring $w[i..i{+}\mu(i){-}1]$. Clearly $\mu'(h) \le \mu(h)$. By Lemma~\ref{Unbordered}, $w[i..i{+}\mu(i){-}1]$ is unbordered and hence its minimal period is $\mu(i)$. By Theorem~\ref{CritFact}, there is $h \in [i..i{+}\mu(i))$ such that $\mu'(h) = \mu(i)$. But for each $h \in [i..j]$, we have $\mu'(h) < \mu(i)$ and moreover, for each $h \in (j..i{+}\mu(i))$, $\mu'(h) \le \mu(j) < \mu(i)$ because the local period $\mu'(j)$ is right external with respect to $w[i..i{+}\mu(i){-}1]$, a contradiction.
\end{proof}

Choose a contiguous subsequence $S' = \{i_1, i_2, \ldots, i_z\}$ of $S$ such that $\mu(i_j) \le q$ for each $j \in [1..z]$ and $\sum_{i \in S'\cap\hat{S}} \mu(i) = T(q)$. As above, we associate with each $i_j$ the values $c_j$ and $r_j$ defined in Section~\ref{SectAnalysis}. By an inductive process described in Section~\ref{SectAnalysis}, we construct a subsequence $\{i_{s_j}\}_{j=1}^t$ of $S'$. The following result complements Lemma~\ref{MainLemma}.
\begin{lemma}
Let $h \in (s_t..z]$ and $\mu(i_h) = \mu(i_{s_{t-1}})$. If $\frac{7}3 \mu(i_{s_{t-1}}) > \mu(i_{s_t})$, then for each $h' \in (h..z]$, we have $i_{h'} \notin \hat{S}$.\label{LastLemma}
\end{lemma}
\begin{proof}
We are to show that, informally, Algorithm~2 processes the position $i_h$ in the same manner as it processed $i_{s_{t-1}}$ and the loop in lines~\ref{lst:copy0}--\ref{lst:copyend} copies all required local periods $\mu(i_{h'})$ for $h' \in (h..z]$ to the array $m$ immediately after the computation of $r[i_{s_t}]$. (Thus $i_{h'} \notin \hat{S}$ for $h' \in (h..z]$.)

Denote $a = w[c_{s_{t-1}}..c_{s_{t-1}}{+}\mu(c_{s_{t-1}}){-}1]$ and $b = w[c_{s_{t-1}}{+}\mu(c_{s_{t-1}}) .. c_{s_{t-1}}{-}\mu(c_{s_{t-1}}){+}\mu(i_{s_t}){-}1]$ (see Figure~\ref{fig:xstring}). Note that $\mu(c_{s_{t-1}}) = \mu(i_{s_{t-1}}) = |a|$ and $\mu(c_{s_t}) = \mu(i_{s_{t}}) = |aab|$. Since $\frac{7}3|a| = \frac{7}3 \mu(i_{s_{t-1}}) > \mu(i_{s_t}) = |aab|$, we have $|b| < \frac{1}3|a|$. By Lemma~\ref{Unbordered}, the string $a$ is unbordered. Denote $x = w[i_{s_t}{-}|aab|..c_{s_t}{+}|aab|{-}1]$ (see Figure~\ref{fig:xstring}). Clearly, $x$ is a substring of the infinite string $aab\cdot aab\cdot aab\cdots$ and the length of $x$ is at least $2|aab|$ (recall that $c_{s_t}$ can coincide with $i_{s_t}$). Notice that the distance between $i_{s_t}$ and $c_{s_t}$ can be arbitrarily large.
\begin{figure}[htb]
\center
\includegraphics[scale=0.55]{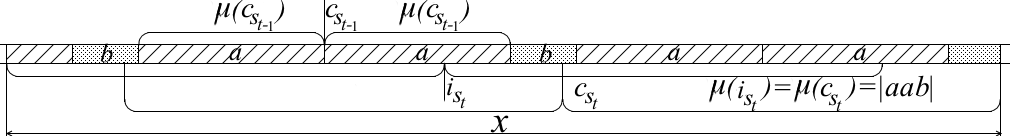}
\caption{The internal structure of the string $x$ from the proof of Lemma~\ref{LastLemma}.}
\label{fig:xstring}
\end{figure}

Without loss of generality, assume that $i_h$ is equal to the leftmost position $i > c_{s_t}$ such that $\mu(i) = \mu(i_{s_{t-1}}) = |a|$. (Since $\{i_1, \ldots, i_z\}$ is a contiguous subsequence of $S$, $i$ is certainly equal to $i_h$ for some $h \in (s_t..z]$.) Obviously $i_h \in (c_{s_t}..c_{s_t}{+}|aab|)$. It follows from the definition of $i_h$ and from Lemma~\ref{MainLemma} that for each $i \in (c_{s_t}..i_h)$, we have $\mu(i) < \frac{2}3|a|$. So, Lemma~\ref{LeftExtBound} implies that $i_h - \mu(i_h) = i_h - |a| < c_{s_t}$. Since by Lemma~\ref{Unbordered} the string $w[c_{s_t}..c_{s_t}{+}|aab|{-}1]$ is unbordered and thus cannot have the period $|a| < |aab|$, we obtain $r_h < c_{s_t} + |aab|$. Thus, the string $w[i_h{-}|a|..r_h]$ is a substring of $x$ (see Figure~\ref{fig:xstring2}). Now we must specify where the position $i_h$ can occur in $x$.
\begin{figure}[htb]
\center
\includegraphics[scale=0.55]{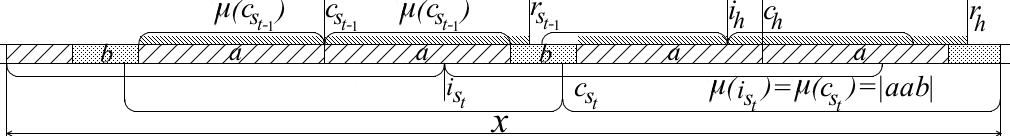}
\caption{A location of $i_h$, $c_h$, and $r_h$ inside $x$ from the proof of Lemma~\ref{LastLemma}.}
\label{fig:xstring2}
\end{figure}

By Lemma~\ref{cjump}, for any $i \in (c_{s_{t-1}}..c_{s_{t-1}}{+}|a|)$, we have $\mu(i) \ne |a|$. Hence $i_h \notin (c_{s_{t-1}}..c_{s_{t-1}}{+}|a|)$. Moreover, since $x$ is a substring of the infinite string $aab\cdot aab\cdot aab\cdots$ and $w[i_h{-}|a|..i_h{+}|a|{-}1]$ is a substring of $x$, in the same way one can prove that $i_h$ does not lie in the segments $(c_{s_{t-1}}{+}|aba|..c_{s_{t-1}}{+}|abaa|)$, $(c_{s_{t-1}}{+}|abaaba|..c_{s_{t-1}}{+}|abaabaa|), \ldots$ (see Figure~\ref{fig:xstring2}), i.e., informally, $i_h$ cannot lie in the right half of an occurrence of $aa$ in $x$.

Suppose $i_h \in [c_{s_{t-1}}{+}|a|..c_{s_{t-1}}{+}|ab|)$. Then, the string $w[i_h{-}|a|..c_{s_{t-1}}{+}|a|]$, which is a suffix of $a$, has an occurrence $v = w[i_h..c_{s_{t-1}}{+}|aa|]$ (see Figure~\ref{fig:centerc}a with $i = i_h$). Since $\mu(i_h) = |a| > |b|$, $v$ overlaps $w[c_{s_{t-1}}{+}|ab|..c_{s_{t-1}}{+}|aba|{-}1] = a$. Thus, $a$ has a nontrivial border, a contradiction. By the same argument, one can show that $i_h$ does not lies in the segments $[c_{s_{t-1}}{+}|abaa|..c_{s_{t-1}}{+}|abaab|)$, $[c_{s_{t-1}}{+}|abaabaa|..c_{s_{t-1}}{+}|abaabaab|), \ldots$; in other words, $i_h$ cannot lie in an occurrence of $b$ in $x$.

We have proved that $i_h$ lies in the left half of an occurrence of $aa$ in $x$, precisely, in one of the segments $[c_{s_{t-1}}{+}|ab|..c_{s_{t-1}}{+}|aba|]$, $[c_{s_{t-1}}{+}|abaab|..c_{s_{t-1}}{+}|abaaba|], \ldots$. Figure~\ref{fig:xstring2} illustrates the case $i_h \in [c_{s_{t-1}}{+}|ab|..c_{s_{t-1}}{+}|aba|]$; all other cases are similar. First, we show that $c_h$ is equal to $c_{s_{t-1}} + |aba|$, i.e., $c_h$ is the center of an occurrence of $aa$ in $x$ (see Figure~\ref{fig:xstring2}). Obviously, the string $w[i_h{-}|a|..c_{s_{t-1}}{+}|abaa|{-}1]$ has the period $|a|$ and therefore $c_{s_{t-1}} + |abaa| \le r_h$. The strings $w[c_{s_{t-1}}{+}|ab|..r_h{-}1]$ and $w[c_{s_{t-1}}{-}|a|..r_{s_{t-1}}{-}1]$ are similar: they both have the period $|a|$, and $w[r_h] \ne w[r_h{-}|a|]$ and $w[r_{s_{t-1}}] \ne w[r_{s_{t-1}}{-}|a|]$. Note that the starting positions of these strings differ by $|aab|$. Furthermore, since $r_h < c_{s_t} + |aab|$, the strings $w[c_{s_{t-1}}{+}|ab|..r_h]$ and $w[c_{s_{t-1}}{-}|a|..r_{s_{t-1}}]$ both are substrings of $x$ and hence they are equal because $x$ has the period $|aab|$. Now since $w[c_{s_{t-1}}{+}|ab|..r_h]$ is a suffix of $w[i_h{-}|a|..r_h]$, it is straightforward that $c_h = c_{s_{t-1}} + |aba|$.

To finish the proof, it suffices to show that Algorithm~2 does not compute explicitly the local periods at the positions $i_{h+1}, i_{h+2}, \ldots,i_z$ but obtains those local periods from the array $m$. For this purpose, let us first prove that for each $h' \in (h..z]$, the string $w[i_{h'}{-}\mu(i_{h'})..i_{h'}{+}\mu(i_{h'}){-}1]$ is a substring of $w[c_h{-}|a|..c_h{+}|a|{-}1]$. This fact implies that, in a sense, after the processing of the position $c_h$ Algorithm~2 is in a situation that locally resembles the situation in which the algorithm was after the processing of the position $c_{s_{t-1}}$ (see Figure~\ref{fig:clarification}), i.e., Algorithm~2 examines exactly the same positions $i_{h+1}, i_{h+2}, \ldots, i_z$ shifted by $\delta = c_h - c_{s_{t-1}}$ or, more formally, $i_{s_{t-1}{+}1} = i_{h+1} - \delta, i_{s_{t-1}{+}2} = i_{h+2} - \delta, \ldots, i_{s_{t-1}{+}z{-}h} = i_z - \delta$.
\begin{figure}[htb]
\center
\includegraphics[scale=0.55]{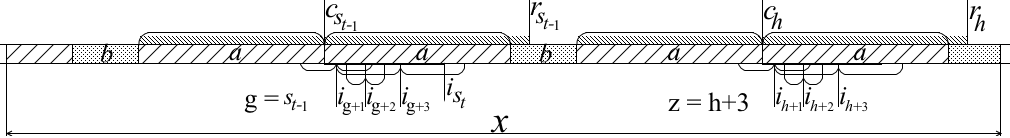}
\caption{Local similarities between $c_{s_{t-1}}$ and $c_h$ in the proof of Lemma~\ref{LastLemma}; for brevity, denote $g = s_{t-1}$. Here $z = h + 3$.}
\label{fig:clarification}
\end{figure}

Let $i$ be the leftmost position from $(c_h..c_h{+}|a|)$ such that $\mu(i) \ge \mu(c_h)$. Lemmas~\ref{Break} and~\ref{cjump} imply that such position always exists and $\mu(i) \ge 2\mu(c_h) = |aa|$. Since $i \in (c_{s_t}..c_{s_t}{+}\mu(c_{s_t}))$ and $|aa| > \frac{1}2|aab| = \frac{1}2\mu(c_{s_t})$, it follows from Lemmas~\ref{LeftExt} and~\ref{cjump} that $\mu(i) \ge 2\mu(c_{s_t})$. Hence, by the definition of the subsequence $\{i_{s_j}\}_{j=1}^t$, we have $i > i_z$. Thus, for each $h' \in (h..z]$, we have $\mu(i_{h'}) < \mu(c_h)$ and $i_{h'} \in (c_h..i)$. Therefore, by Lemma~\ref{NotRight}, the string $w[i_{h'}{-}\mu(i_{h'})..i_{h'}{+}\mu(i_{h'}){-}1]$ is a substring of $w[c_h{-}|a|..c_h{+}|a|{-}1]$.

Suppose $i_{s_t} \in \hat{S}$. Summing up the established facts, we obtain that since $\delta = c_h - c_{s_{t-1}}$ is a multiple of $\mu(i_{s_t}) = |aab|$, the loop in lines~\ref{lst:copy0}--\ref{lst:copyend} performed immediately after the computation of the local period at the position $i_{s_t}$ in line~\ref{lst:computelocper} copies $m[i_{h+1}{-}\delta], m[i_{h+2}{-}\delta], \ldots, m[i_z{-}\delta]$, which are certainly filled with nonzero values, to $m[i_{h+1}], m[i_{h+2}], \ldots, m[i_z]$, respectively. Thus, Algorithm~2 does not compute explicitly the local periods at the positions $i_{h+1}, i_{h+2}, \ldots, i_z$.

Suppose $i_{s_t} \notin \hat{S}$, i.e., $m[i_{s_t}]$ and $r[i_{s_t}]$ are nonzero at the time the algorithm reaches $i_{s_t}$. It follows from Algorithm~2 that the values $m[i_{s_t}]$ and $r[i_{s_t}]$ are obtained from values $m[i']$ and $r[i']$ for some position $i' < i_{s_t}$ such that $w[i'{-}m[i']..r[i']] = w[i_{s_t}{-}m[i_{s_t}]..r[i_{s_t}]]$. Suppose $i' \in \hat{S}$. Thus, when Algorithm~2 had calculated $\mu(i')$, it passed through the positions $i_{s_t{+}1}{-}\delta, i_{s_t{+}2}{-}\delta, \ldots, i_z{-}\delta$, where $\delta = i_{s_t} - i'$, stored the corresponding local periods in $m[i_{s_t{+}1}{-}\delta], m[i_{s_t{+}2}{-}\delta], \ldots, m[i_z{-}\delta]$, and then copied those values to $m[i_{s_t{+}1}], m[i_{s_t{+}2}], \ldots, m[i_z]$, respectively, when copied $m[i']$ to $m[i_{s_t}]$. Finally, suppose $i' \notin \hat{S}$. By an obvious induction, one can prove that in this case $m[i_{s_t{+}1}{-}\delta], m[i_{s_t{+}2}{-}\delta], \ldots, m[i_z{-}\delta]$ are also filled with correct values and thus the same argument shows that $m[i_{s_t{+}1}], m[i_{s_t{+}2}], \ldots, m[i_z]$ are eventually set to nonzero values.
\end{proof}

Suppose $t > 1$ and $\frac{7}3\mu(i_{s_{t-1}}) \le \mu(i_{s_t})$. As in Section~\ref{SectAnalysis}, $T(q)$ is determined by the recursion (\ref{eq:teq}). Let us estimate the sum of the terms from the parentheses of $T(\ldots)$ in the right hand side of (\ref{eq:teq}). Since $\mu(i_{s_{t-1}}) \le \frac{3}{7} \mu(i_{s_t})$, we have $\frac{1}2 \mu(i_{s_1}) + \cdots + \frac{1}2 \mu(i_{s_t}) \le \frac{3}{7}\mu(i_{s_t})(\frac{1}2 + \frac{1}{2^2} + \frac{1}{2^3}~+~\cdots) + \frac{1}2\mu(i_{s_t}) \le \frac{3}{7} q + \frac{1}2q = \frac{13}{14}q$.

Suppose $t > 1$, $\frac{7}3\mu(i_{s_{t-1}}) > \mu(i_{s_t})$. Let $h$ be the minimal number from $(s_t..z]$ such that $\mu(i_h) = \mu(i_{s_{t-1}})$ (if it does not exist, assume that $h = z$). By the definition of the subsequence $\{i_{s_j}\}_{j=1}^t$, we have $i_h \in (c_{s_t}..c_{s_t}{+}\mu(c_{s_t}))$. Lemma~\ref{MainLemma} implies that $\mu(i) < \frac{2}3\mu(i_{s_{t-1}}) \le \frac{1}3 \mu(i_{s_t})$ for each $i \in (c_{s_t}..i_h)$. Further, by Lemma~\ref{LastLemma}, we have $i_{h'} \notin \hat{S}$ for each $h' \in (h..z]$ and thus we can ignore these positions in our analysis. So, combining Lemmas~\ref{LeftExt},~\ref{LeftExtBound},~\ref{MainLemma},~\ref{LastLemma}, one can deduce the following recursion:

\begin{equation}
T(q) \le \sum_{j=1}^t \mu(i_{s_j}) + \mu(i_h) + T\left(\frac{1}2\mu(i_{s_1})\right) + \cdots + T\left(\frac{1}2\mu(i_{s_{t-1}})\right) + T\left(\frac{1}3\mu(i_{s_t})\right)\enspace.
\label{eq:main_recustion}
\end{equation}

Let us estimate the sum of the terms from the parentheses of $T(\ldots)$ in the right hand side of (\ref{eq:main_recustion}). Since $\sum_{j=1}^{t-1} \mu(i_{s_j}) \le q$, we have $\frac{1}2 \mu(i_{s_1}) + \cdots + \frac{1}2 \mu(i_{s_{t-1}}) + \frac{1}3 \mu(i_{s_t}) \le \frac{1}2q + \frac{1}3 q = \frac{5}6 q$. Clearly, the sum $\sum_{j=1}^t \mu(i_{s_j}) + \mu(i_h)$ is bounded by~$3q$.

Finally, in the case $t = 1$ we have, by Lemmas~\ref{LeftExt} and~\ref{LeftExtBound}, $T(q) \le \mu(i_{s_1}) + T(\frac{1}2\mu(i_{s_1}))$. Obviously, $\frac{1}2\mu(i_{s_1})$, the term from the parentheses of $T(\ldots)$, is less than or equal to $\frac{1}2q$.

Putting everything together, it is easy to see that $T(q)$ is determined by the recursion $T(q) \le 3q + \sum_{j=1}^{r} T(q_j)$ for some terms $\{q_j\}_{j=1}^r$ such that $\sum_{j=1}^r q_j \le \alpha q$, where $\alpha = \min\{\frac{13}{14}, \frac{5}6, \frac{1}2\} < 1$. It is well known that such recursion has the solution $T(q) \le 3q + \alpha 3q + \alpha^2 3q~+~\cdots = \frac{3q}{1 - \alpha} = O(q)$. Thus, the above analysis of Algorithm 2 proves the following theorem.

\begin{theorem}
There is a linear time and space algorithm finding the leftmost critical point of a given string on an arbitrary unordered alphabet.
\end{theorem}

\section{Conclusion}\label{SectConclusion}

We have shown that the problems of the computation of a critical factorization on unordered and ordered alphabets both have linear time solutions. This is in contrast with the seemingly related problem of finding repetitions in strings (squares, in particular) for which it is known that in the case of unordered alphabet one cannot even check in $o(n\log n)$ time whether the input string of length $n$ contains some repetitions while in the case of ordered alphabet there are fast $o(n\log n)$ time checking algorithms (see~\cite{Kosolobov4,Kosolobov,Kosolobov2,MainLorentz}). The search of similarities between those problems was actually our primary motivation for the present work although our result shows that the restriction to the case of unordered alphabets does not add considerable computational difficulties to the problem of the calculation of a critical factorization unlike the problem of finding repetitions, so, they are not similar in this aspect.

As a byproduct, we have obtained the first generalization of the constant space string matching algorithm of Crochemore and Perrin~\cite{CrochemorePerrin} to unordered alphabets. However, this generalization requires nonconstant space in the preprocessing step. So, it is still an open question to find a linear time and \emph{constant space} algorithm computing a critical factorization (not necessarily the leftmost one) of a given string on an arbitrary unordered alphabet. Using such tool, one can possibly obtain a constant space string matching algorithm that is simpler and faster than the well-known algorithm of Galil and Seiferas~\cite{GalilSeiferas}.

\noindent\textbf{Acknowledgement.} The author would like to thank Arseny M. Shur for helpful discussions and the invaluable help in the preparation of this paper.

\section*{References}
\bibliography{cf}

\end{document}